\documentclass[journal]{IEEEtran}

\ifCLASSINFOpdf
\else
   \usepackage[dvips]{graphicx}
\fi
\usepackage{url}

\usepackage{graphicx}
\usepackage{amsmath, amsfonts, amssymb}
\usepackage{cite}
\allowdisplaybreaks 

\usepackage{algorithm}
\usepackage{algorithmic}

\usepackage{xcolor}

\usepackage{amsthm}
\newtheorem{theorem}{Theorem}
\newtheorem{corollary}{Corollary}
\newtheorem{proposition}{Proposition}
\newtheorem{definition}{Definition}
\newtheorem{lemma}{Lemma}

\def \x {{\mathbf{x}}}

\def\alphab {{\boldsymbol{\alpha}}}
\DeclareMathOperator*{\argmax}{arg\,max}

\begin{document}

\title{Temporal Parallelization of Inference in Hidden Markov Models}

\author{Syeda Sakira Hassan, \IEEEmembership{Member, IEEE,} Simo S\"arkk\"a, \IEEEmembership{Senior Member, IEEE,} and \'Angel F. Garc\'ia-Fern\'andez %
\thanks{S. Hassan and S. Särkkä are with the Department of Electrical Engineering and
Automation, Aalto University, 02150 Espoo, Finland (emails: {syeda.s.hassan}@aalto.fi,
{simo.sarkka}@aalto.fi).} %
\thanks{\'Angel F. Garc\'ia-Fern\'andez is with the Department of Electrical Engineering and Electronics, University of Liverpool, Liverpool L69 3GJ, United Kingdom (email: angel.garcia-fernandez@liverpool.ac.uk). He is also with  the ARIES Research Centre, Universidad Antonio de Nebrija,  Madrid, Spain.} 
\thanks{Digital Object Identifier 10.1109/TSP.2021.3103338}}

\markboth{}
{Shell \MakeLowercase{\textit{et al.}}: Bare Demo of IEEEtran.cls for IEEE Journals}
\maketitle

\begin{abstract}
This paper presents algorithms for the parallelization of inference in hidden Markov models (HMMs). In particular,  we propose a parallel forward-backward type of filtering and smoothing algorithm as well as a parallel Viterbi-type  maximum-a-posteriori (MAP) algorithm. We define associative elements and operators to pose these inference problems as all-prefix-sums computations and parallelize them using the parallel-scan algorithm. The advantage of the proposed algorithms is that they are computationally efficient in HMM inference problems with long time horizons. We empirically compare the performance of the proposed methods to classical methods on a highly parallel graphics processing unit (GPU).
\end{abstract}

\begin{IEEEkeywords}
Parallel forward-backward algorithm, parallel sum-product algorithm, parallel max-product algorithm, parallel Viterbi algorithm.
\end{IEEEkeywords}

\IEEEpeerreviewmaketitle

\section{Introduction}

\IEEEPARstart{H}{idden} Markov models (HMMs) have gained a lot of attention due to their simplicity and a broad range of applications \cite{rabiner1989tutorial, cappe2006inference, song2013maxproduct, ueng2013trellis, mor2020systematic}. Successful real-world application areas of HMMs include speech recognition, convolutional code decoding, target tracking and localization, facial expression
recognition, gene prediction, gesture recognition, musical
composition, and bioinformatics~\cite{rabiner1989tutorial, juang1991hidden, salah2007hidden, krogh2001predicting, pan2015viterbiDecoder, yang2018viterbiDecoding, long2011probabilistic}. An HMM is a statistical model that provides a simple and flexible framework that can be used to express the conditional independence and joint distributions using graph-like structures. An HMM can be thought of as a specific form of a probabilistic graphical model consisting of two components: a structural component that defines the \textit{edges} and a parametric component that encodes \textit{potentials} associated with the edges in the graph. A graphical representation of an HMM is shown in Fig.~\ref{fig:hmm}. An HMM is a doubly stochastic process, where the underlying stochastic process (light gray-colored nodes) can be only observed through another stochastic process (dark gray-colored nodes)~\cite{rabiner1989tutorial, viesti2019}. A primary task in graphical models is to perform inference, that is, to compute marginals or the most likely values of the unobserved modes given the observed modes.

\begin{figure}
\centerline{\includegraphics[width=\columnwidth]{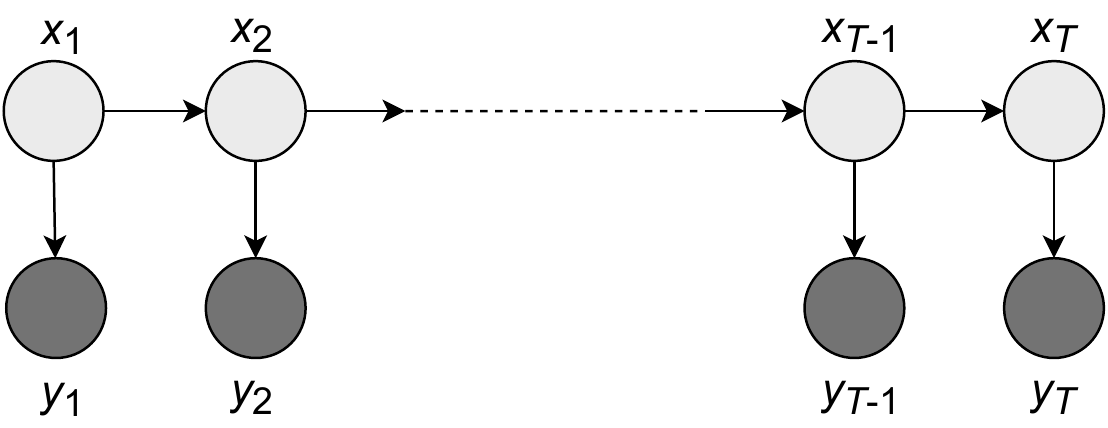}}
\caption{An HMM diagram. Observed nodes are shaded in dark gray, whereas unobserved nodes are shaded in light gray.}
\label{fig:hmm}
\end{figure}

If the HMM model has $D$ states and the length of the sequence is $T$, then there are $D^T$ possible state sequences. For the forward-backward algorithm, whose objective is to find the marginal distributions, and for the Viterbi algorithm, whose objective is to find the most likely sequence (the Viterbi path), the time complexity is $O(D^2T)$~\cite{viterbi1967error,rabiner1989tutorial, juang1991hidden}. Some methods to speed up these algorithms via parallelization have appeared in literature. For instance, in order to speed up the inference task, Lifshits et al.~\cite{lifshits2009speeding} used a compression scheme, which utilizes the repetitive patterns in the observed sequences. Sand et al.~\cite{sand2010hmmlib} developed a parallel forward algorithm using single instruction multiple data (SIMD) processors and multiple cores. Nielsen and Sand~\cite{nielsen2011algorithms} presented a parallel reduction on the forward algorithm. Chatterjee and Russell~\cite{chatterjee2012temporally} used the temporal abstraction concept from dynamic programming to speed up the Viterbi algorithm. Using accelerated hardware, a tile-based Viterbi algorithm was proposed by Zhihui et al.~\cite{zhihui2010TileVB}, where the matrix multiplication was done in parallel. Maleki et al.~\cite{maleki2014parallelizing} proposed an optimized method that is able to solve a particular class of dynamic programming problems by tropical semirings. They showed that it can be used to optimize the Viterbi algorithm. Nevertheless, parallelization has not been fully investigated in HMM inference tasks. In this paper, we develop novel parallel versions of the forward-backward algorithm and the Viterbi algorithm that have $O(\log T)$ span complexities.

 Parallel algorithms can take advantage of the computational power of specialized hardware accelerators, such as general purpose \textit{graphics processing units} (GPUs)~\cite{owens2008gpu}, \textit{neural processing units} (NPUs)~\cite{esmaeilzadeh2012npu}, or \textit{tensor processing units} (TPUs)~\cite{jouppi2017datacenter}. These accelerators allow us to perform the parallel computation on a large amount of data, making computation faster and
cheaper and, therefore, economically viable. Among these accelerators, GPUs are the most widely used alternatives. GPU architectures enable us to harness the massive parallelization of general-purpose algorithms~\cite{luebke2008cuda}. Table~\ref{tab:previous_work} summarizes relevant works on the HMM inference tasks, which were implemented on GPUs. However, these works were optimized particularly for speech recognition tasks and did not explore the full capabilities of parallelism.
In order to utilize parallelism, sequential algorithms need to be reformulated in terms of primitive operations that enable us to perform the parallel execution on parallel platforms. For example, the parallel-scan algorithm~\cite{Blelloch:1989, Blelloch:1990} can be used to run computations in parallel provided that they can be formulated in terms of binary associative operators.

\begin{table}[!ht]
    \centering
    \caption{Previous works on the HMM inference task using GPUs. The notation `-' means that the value was not mentioned in the referenced article. }
    \begin{tabular}{l|c|c|c}
         Algorithm & States & Observations & Speed improvement \\
         \hline
         Forward-backward~\cite{li2009fast} & 8       & 200        & 3.5x  \\
         Forward~\cite{liu2009cuhmm}            & 512   & 3-10       & 880x\\
         Baum-Welch~\cite{liu2009cuhmm}         & 512   & 3-10       & 180x\\
         Viterbi~\cite{zhang2009implementation} &   -    &2000-3000   & 3x\\
         Forward~\cite{hymel2011massively}      & 4000  &      1000      & 4x\\
         Baum-Welch~\cite{hymel2011massively}   & 4000  &     1000     & 65x
    \end{tabular}
    \label{tab:previous_work}
\end{table}

Recently, the parallel-scan algorithm has been used in parallel Bayesian filtering and smoothing algorithms to speed up the sequential computations via parallelization~\cite{sarkka2020temporal, yaghoobi2021parallel}. Although this framework is applicable to HMMs as well, the difference to our proposed framework is that here we formulate the inference problem in terms of potentials. This results in a different backward pass to compute the posterior marginals which corresponds to a two-filter smoother formulation, whereas the formulation in Ref.~\cite{sarkka2020temporal,yaghoobi2021parallel} is a Rauch--Tung--Striebel type of smoother~\cite{sarkka2013bayesian}. In this article, we also present a parallel Viterbi-type maximum-a-posteriori (MAP) algorithm which has not been explored before.

The main contribution of this paper is to present a parallel framework to perform  the HMM inference tasks efficiently by using the parallel-scan algorithm. In particular, we formulate the sequential operations of sum-product and max-product algorithms as binary associative operations. This formulation allows us to construct parallel versions of the forward-backward algorithm and the Viterbi algorithm with $O(\log T)$ span complexities. For the latter algorithm, we propose two alternative parallelization approaches: a path-based and a forward-backward based formulation. We also empirically evaluate the computational speed advantage of these methods on a GPU platform.

The structure of the paper is the following. In Section~\ref{sec:problem_formulation}, we define the HMM inference problems using a probabilistic graph. Then, in Section~\ref{sec:sum-product}, we review the classical sum-product and forward-backward algorithms, introduce the elements and associative operations for parallel operations, and formulate the parallel-scan algorithm on these elements and operators. Next, in Section~\ref{sec:max-product}, we show that we can apply the parallel-scan algorithm to the max-product algorithm, which results in a parallel version of the Viterbi algorithm. In Section~\ref{sec:extensions}, we discuss extensions and generalizations of the methodology, and in Section~\ref{sec:experimental_results}, we present experimental results for both parallel sum-product and parallel max-product based algorithms. Section~\ref{sec:conclusion} concludes the article.

\section{Problem Formulation}
\label{sec:problem_formulation}
Assume that we have $T$ discrete random variables $\x = (x_1, \ldots, x_T)$ which correspond to nodes in a probabilistic graph and $N$ potential functions $\psi_1, \ldots, \psi_N$ which define the cliques of the graph \cite{Koller:2009, Rangan2017messagePassing, weiss2000correctness}. We also assume each variable $x_t$ takes values in the set $\mathcal{X} = \{ 1,\ldots,D \}$, where $D$ represents the number of states.  Each potential is a function of a subset of $x_t$'s, defined by multi-indices $\alphab_1, \ldots, \alphab_T$ with elements $\alphab_t = (\alpha_{t,1}, \ldots, \alpha_{t, |\alphab_t|})$. We denote the subset as $\x_{\alphab_t}$. The joint distribution of the random variables has the representation
\begin{equation}
    p(\x) = \frac{1}{Z} \prod_{t=1}^T \psi_t\left( \x_{\alphab_t} \right),
\label{eq:joint_dist}
\end{equation}
where $Z = \sum_{\x} \prod_{t=1}^T \psi_t\left( \x_{\alphab_t} \right)$ is the normalization constant, also known as the partition function~\cite{Koller:2009}. A typical inference task is the computation of all the marginals
\begin{equation}
    p(x_k) = \frac{1}{Z} \sum_{\x \backslash x_k} \prod_{t=1}^T \psi_t\left( \x_{\alphab_t} \right),
\label{eq:marginals}
\end{equation}
where $\x \backslash x_k = (x_1, \ldots, x_{k-1}, x_{k+1}, \ldots, x_T)$, that is, the summation is performed over all the variables except $x_k$. Another inference task is the computation of the similar quantity for the maximum
\begin{equation}
    p^*(x_k) = \max_{\x \backslash x_k} \prod_{t=1}^T \psi_t\left( \x_{\alphab_t} \right),
\label{eq:maxima}
\end{equation}
where the maximum is computed with respect to all variables in $\x$ except $x_k$.

Let us now consider inference problems in a hidden Markov model (HMM) of the form
\begin{subequations}
\begin{align}
    x_k &\sim p(x_k \mid x_{k-1}), \\
    y_k &\sim p(y_k \mid x_k),
    \label{eq:ssm}
\end{align}
\end{subequations}
where $\sim$ stands for ``is distributed as". Here, we assume that the sequence $x_1, \ldots, x_T$ is Markovian with transition probabilities $p(x_k \mid x_{k-1})$, and the observations $y_k$ are conditionally independent given $x_k$ with likelihoods $p(y_k \mid x_k)$. Furthermore, we have a prior $x_1 \sim p(x_1)$. We are interested in computing the smoothing posterior distributions $p(x_k \mid y_1, \ldots, y_T)$ for all $k = 1, \ldots, T$ as well as in computing the MAP estimate (i.e., the Viterbi path) of $p(\x)$. These can be computed using sequential algorithms in a linear computational time \cite{viterbi1967error, rabiner1986introduction,sarkka2013bayesian}. However, our aim is to reduce this computational time by using parallelization in the temporal domain.

We can express the inference problem in Eq.~\eqref{eq:joint_dist} by defining
\begin{subequations}
\begin{align}
    \psi_1 (x_1) &= p(y_1 \mid x_1) \, p(x_1), \\
    \psi_k(x_{k-1}, x_{k}) &= p(y_k \mid x_k) \, p(x_k \mid x_{k-1}), \quad \text{ for } k > 1.
\end{align}
\end{subequations}
With these definitions the joint distribution in Eq.~\eqref{eq:joint_dist} takes the form
\begin{equation}
    p(\x) = \frac{1}{Z} \psi_1 (x_1) \prod_{t=2}^T \psi_t(x_{t-1}, x_{t}).
\label{eq:joint_dist_markov}
\end{equation}
In this Markovian case, each potential is a function of the neighboring nodes $x_{k-1}, x_k$, that is, we have $\alphab_k = (k-1,k)$ for $k > 1$ and $\alphab_1 = (1)$. 

Typically, the inference task in the HMM is either to compute the marginals $p(x_k \mid y_1, \ldots, y_T)$ or to find the MAP sequence $x^*_{1:T}$, the Viterbi path. In the potential function formulation, the smoothing distribution is given by Eq.~\eqref{eq:marginals} and the maximum of Eq.~\eqref{eq:maxima} provides the Viterbi path at step $k$. Both of these correspond to the following kinds of general operations on Eq.~\eqref{eq:joint_dist_markov}:
\begin{equation}
    F(x_k) = \frac{1}{Z}\, \mathrm{OP}_k \left(\psi_1 (x_1) \prod_{t=2}^T \psi_t(x_{t-1}, x_{t})\right),  k = 1, \ldots, T,
\label{eq:general_operation}
\end{equation}
where $\mathrm{OP}_k(.)$ is a sequence of operations applied to all elements but $x_k$, such as $\sum_{\x \backslash x_k}$ or $\max_{\x \backslash x_k}$, resulting in a function of $x_k$. One way to compute these operations in Eq.~\eqref{eq:general_operation} is to use sum-product or max-product algorithms~\cite{pearl1988probabilistic, shafer1990probabilityPropagation, Bishop:2006}, for summation and maximization operations, respectively. However, in the following section we show that, since the summation and the maximum operations are associative, we can use parallel-scan algorithm for parallelizing these computations. The same principle would also apply to any other binary associative operation, but here we specifically concentrate on these two operations.

\section{Parallel sum-product algorithm for HMMs}
\label{sec:sum-product}

In this section, we present a parallel formulation of the sum-product algorithm. In Section~\ref{subsec:classical_sum_product_algorithm}, we review the classical sum-product algorithm. In Section~\ref{subsec:parallel_scan_algorithm}, we revisit the parallel-scan algorithm. In Section~\ref{subsec:sum_product_associative}, we show the decomposition of the sum-product algorithm in terms of the binary associative operations. Finally, in Section~\ref{sec:par_sum_product}, we propose the parallel sum-product algorithm with $O(\log T)$ span complexity.

\subsection{Classical sum-product algorithm}
\label{subsec:classical_sum_product_algorithm}

The sum-product algorithm  can be used to find the marginal distributions of all variables in a graph~\cite{cowell2006probabilistic, frey1998graphical}. We first define the forward potential as
\begin{equation}
\begin{aligned}
\psi_{1, k}^f(x_k) &= \sum_{x_{1:k-1}} \psi_1(x_1) \prod_{t=2}^k \psi_{t-1, t}(x_{t-1}, x_{t})  \\
\end{aligned}
\label{eq:definition_f_potentials}
\end{equation}
and the backward potential as
\begin{equation}
\begin{aligned}
\psi_{k, T}^b(x_k) &= \sum_{x_{k+1:T}} \prod_{t=k}^T \, \psi_{t, t+1}(x_{t}, x_{t+1}).
\end{aligned}
\label{eq:definition_b_potentials}
\end{equation}
Here, we have defined $\psi_{T, T+1}(x_T, x_{T+1}) = 1$ and denoted $\sum_{x_{1:k}} = \sum_{x_1} \ldots \sum_{x_k}$.  We can now express the marginal distribution $p(x_k)$, which corresponds to summation operation in Eq.~\eqref{eq:general_operation}, as a normalized product of the forward and backward potentials
\begin{equation}
\begin{aligned}
    p(x_k) 
    &= \frac{1}{Z_k} \, \psi_{1,k}^f(x_k) \, \psi_{k,T}^b(x_k),
\end{aligned}
\label{eq:smoothing_dist_extended}
\end{equation}
where $Z_k = \sum_{x_k} \psi_{1,k}^f(x_k) \, \psi_{k,T}^b(x_k)$. If we want to compute all the marginals $p(x_{1})$, $p(x_{2}), \ldots, p(x_T)$, then we need to compute all the terms $\psi_{1,k}^f(x_k)$ and $\psi_{k,T}^b(x_k)$ and combine them. It turns out that we can compute these forward and backward potentials in $O(D^2 \, T)$ steps using Algorithm~\ref{alg:fb_algorithm}, which is an instance of a sum-product algorithm.

\begin{algorithm}[H]
\caption{The classical sum-product algorithm for computing the forward and backward potentials.}
\label{alg:fb_algorithm}
\begin{algorithmic}[1]
  \REQUIRE The potentials $\psi_k(\cdot)$ for $k=1,\ldots,T$.
  \ENSURE The forward and backward potentials $\psi_{1,k}^f(x_k)$ and $\psi_{k,T}^b(x_k)$ for $k=1,\ldots,T$.
  \STATE // Forward pass:
  \STATE {$\psi_{1,1}^f(x_1) = \psi_1(x_1)$}  \COMMENT{Initialization}
  \FOR[Sequentially]{$k \gets 2$ \textbf{to} $T$} 
  \STATE {$\psi_{1,k}^f(x_k)
      = \sum_{x_{k-1}} \psi_{1,k-1}^f(x_{k-1}) \, \psi_{k-1, k}(x_{k-1}, x_k)$}
  \ENDFOR
  \STATE // Backward pass:
  \STATE {$\psi_{T,T}^b(x_T) = 1$} \COMMENT{Initialization}
  \FOR[Sequentially]{$k \gets T-1$ \textbf{to} $1$} 
     \STATE $\psi_{k,T}^b(x_k) = \sum_{x_{k+1}} \psi_{k, k+1}(x_k, x_{k+1}) \,  \psi_{k+1,T}^b(x_{k+1})$
  \ENDFOR%
\end{algorithmic}
\end{algorithm}

It should be noted that the belief propagation algorithm ~\cite{pearl1988probabilistic}, operating on a Bayesian network, corresponds to the sum-product algorithm in a factor graph with similar factorization. The forward algorithm of the HMM model is equivalent to filtering, whereas the backward algorithm corresponds to the backward pass in two-filter smoothing~\cite{sarkka2013bayesian}.

\subsection{Parallel-scan algorithm}
\label{subsec:parallel_scan_algorithm}

In this section, we revisit the parallel-scan algorithm. 
The parallel-scan algorithm~\cite{Blelloch:1989, ladner1980parallel} is a general framework to compute the generalized all-prefix-sums of associative operators in parallel. It was originally designed for computing all-prefix-sums~\cite{ladner1980parallel} using the summation operator. Later, it was generalized to other associative operators~\cite{Blelloch:1989, Blelloch:1990} and became a fundamental building block for solving many sequential algorithms in parallel including sorting, linear programming, and graph algorithms. 

The generalized all-prefix-sums operation for an operator $\otimes$ is defined as follows.

\begin{definition}
For a sequence of $T$ elements $(a_1, a_2, \ldots, a_T)$ and a binary associative operator $\otimes$, the all-prefix-sums operation computes the following sequence of length $T$:
\begin{align}
&(a_1, a_1 \otimes a_2,  \cdots, a_1 \otimes a_2 \otimes \ldots \otimes a_T).  
\end{align}
\label{def:all_prefix_sum}
\end{definition}
Similarly to the above, we can define the reversed all-prefix-sums as follows.
\begin{definition}
For a sequence of $T$ elements $(a_1, a_2, \ldots, a_T)$ and a binary associative operator $\otimes$, the reversed all-prefix-sums operation computes the following sequence of length $T$:
\begin{align}
&(a_1 \otimes a_2 \otimes \cdots \otimes a_T, \ldots, a_{T-1} \otimes a_T,  a_T).  
\end{align}
\label{def:rev_all_prefix_sum}
\end{definition}
The all-prefix-sums operation can be computed through the parallel-scan algorithm~\cite{Blelloch:1989, Blelloch:1990} in $O(\log T)$ span complexity. 
This algorithm computes only the non-reversed all-prefix-sums. However, it is also possible to compute the reversed all-prefix-sums. This can be achieved by reversing the inputs before performing the parallel-scan algorithm and reversing the outputs after the operation is performed. In addition to these, we also need to reverse the operation inside the algorithm. A pseudocode for the parallel-scan is provided by Algorithm~\ref{alg:parprefix}. 

\begin{algorithm}[H]
\begin{algorithmic}[1]
  \REQUIRE The elements $a_k$ for $k=1,\ldots,T$ and the operator $\otimes$.
  \ENSURE The prefix sums in $a_k$ for $k=1,\ldots,T$. 
 \STATE // Save the input:
  \FOR[Compute in parallel]{$i\gets1$ \textbf{to} $T$}
      \STATE $b_i \gets a_i$
  \ENDFOR
  \STATE // Up sweep:
  \FOR{$d\gets0$ \textbf{to} $\log_2 T - 1$} 
    \FOR[Compute in parallel]{$i\gets0$ \textbf{to} $T-1$ \textbf{by} $2^{d+1}$}
      \STATE $j \gets i + 2^d$
      \STATE $k \gets i + 2^{d+1}$
      \STATE $a_k \gets a_j \otimes a_k$
    \ENDFOR
  \ENDFOR
  \STATE $a_T \leftarrow 0$ \COMMENT{Here, $0$ is actually the neutral element for $\otimes$}
  \STATE // Down sweep:
  \FOR{$d\gets\log_2 T - 1$ \textbf{to} $0$}
    \FOR[Compute in parallel]{$i\gets0$ \textbf{to} $T-1$ \textbf{by} $2^{d+1}$}
      \STATE $j \gets i + 2^d$
      \STATE $k \gets i + 2^{d+1}$
      \STATE $t \gets a_j$
      \STATE $a_j \gets a_k$
      \STATE $a_k \gets a_k \otimes t$
    \ENDFOR
  \ENDFOR
  \STATE // Final pass:
  \FOR[Compute in parallel]{$i\gets1$ \textbf{to} $T$}
       \STATE $a_i \gets a_i \otimes b_i$
  \ENDFOR
\end{algorithmic}
\caption{The parallel-scan algorithm for in-place transformation of the sequence $(a_k)$ into its all-prefix-sums in $O(\log T)$ span complexity. Note that the algorithm in this form assumes that $T$ is a power of $2$, but it can easily be generalized to an arbitrary $T$.}
\label{alg:parprefix}
\end{algorithm}

The fundamental idea of the parallel-scan algorithm is to reorder the computations by using the associative property of the operator so that the resulting independent subproblems can be computed in parallel. The parallel-scan algorithm consists of two parts: \textit{up-sweep} and \textit{down-sweep}. The algorithm can be thought of as two traversals in a balanced binary tree. The first pass is the up-sweep that starts from leaves and ends at the root. The second pass is the down-sweep in reverse direction, which starts from the root and ends at leaves. Then, an additional pass is used to form the final result. Because a binary tree with $T$ leaves has the depth of $\log T$, the span complexity of the algorithm is $O(\log T)$. Our aim is to define the associative operators and elements corresponding to the sum-product and max-product algorithms and use the parallel-scan framework to parallelize the computations.

\subsection{Sum-product algorithm in terms of associative operations}
\label{subsec:sum_product_associative}

We can formulate the sum-product algorithm in a more abstract form by defining a general element $a_{i:k}$ and considering a binary associative operator $\otimes$ such that\cite{Blelloch:1990}
\begin{equation}
\begin{aligned}
a_{i:k} &= a_{i:j} \otimes a_{j:k}, \,\text{for}\,\, i < j < k. \\
\end{aligned}
\end{equation}
As it is shown in the following, the computation of the forward and backward terms reduces to computing $a_{0:k}$ and $a_{k:T+1}$.

\begin{definition}
\label{def:filtering_smoothing_operator}
We define an element $a_{i:k}$ recursively as follows. We have
\begin{equation}
\begin{aligned}
   a_{0:1} &= \psi_1(x_1),\\
   a_{k-1:k} &= \psi_k(x_{k-1}, x_k), \\
   a_{T:T+1} &= 1.
\end{aligned}
\label{eq:definition}
\end{equation}
For notational convenience and to enhance readability, we also define
\begin{equation}
\begin{aligned}
   \psi_{0,1}(x_0,x_1) &\triangleq \psi_1(x_1),\\
   \psi_{k-1,k}(x_{k-1}, x_k) &\triangleq \psi_k(x_{k-1}, x_k), \\
   \psi_{T,T+1}(x_T,x_{T+1}) &\triangleq 1.
\end{aligned}
\label{eq:definition_potentials}
\end{equation}
Now, given two elements $a_{i:j}$ and $a_{j:k}$, the binary associative operator $\otimes$ for  the forward-backward algorithm in an HMM for $0 \le i < j < k$ is 
\begin{align}
 a_{i:j} \otimes a_{j:k} &= \sum_{x_{j}} \psi_{i, j}(x_i, x_j) \psi_{j,k}(x_{j}, x_{k}),
\end{align}
which also implies the following representation for a general element:
\begin{equation}
\begin{aligned}
   a_{i:k} &= \psi_{i,k}(x_i, x_k).
\end{aligned}
\end{equation}
\end{definition}

\begin{lemma}
\label{lem:otimes_assoc}
The operator $\otimes$ is associative.
\end{lemma}

\begin{proof} See Appendix~\ref{sec:proof_of_assoc_fs} for the proof.
\end{proof}

Now, we present two theorems that allow us to formulate the forward-backward algorithm in terms of the associative operations.

\begin{theorem}
The forward potential is given as
\begin{align*}
 a_{0:k} &= \psi_{1,k}^f(x_{k}), \quad k > 0.
\end{align*}
\label{th:theorem_forward_potentials}
\end{theorem}
\begin{proof}
Since the operator is associative, it is enough to prove by induction that
\begin{equation}
\begin{aligned}
a_{0:k} &= a_{0:1} \otimes a_{1:2} \otimes \cdots \\
 &\otimes a_{k-2:k-1} \otimes a_{k-1:k}  =  \psi_{1,k}^f(x_k).
\end{aligned}
\label{eq:theorem1_induction}
\end{equation}
Eq.~\eqref{eq:theorem1_induction} holds for $k=1$ by definition of $a_{0:1}$ in Eq.~\eqref{eq:definition}. That is, $a_{0:1} = \psi_{1}(x_{1})$.
Then, we assume that
\begin{equation}
\begin{aligned}
  &a_{0:k-1} = a_{0:1} \otimes a_{1:2} \otimes \cdots \otimes a_{k-2:k-1}= \psi_{1,k-1}^f(x_{k-1}) 
\end{aligned}
\label{eq:theorem1_induction_case}
\end{equation}
holds. We need to prove that Eq.~\eqref{eq:theorem1_induction} holds for any $k > 1$. By applying the binary operator $\otimes$ with $a_{k-1:k}$ to the left-hand side of Eq.~\eqref{eq:theorem1_induction_case}, we get
\begin{equation*}
\begin{aligned}
&a_{0:k-1} \otimes a_{k-1:k} \\
&=  \sum_{x_{k-1}} \psi_{1,k-1}^f(x_{k-1}) \,  \psi_{k-1,k}(x_{k-1}, x_{k})\\
&= \sum_{x_{k-1}} \bigg[ \left( \sum_{x_{1:k-2}} \psi_1(x_1) \prod_{t=2}^{k-1} \psi_{t-1, t}(x_{t-1}, x_{t}) \right) \\
&\qquad \times \psi_{k-1,k}(x_{k-1}, x_{k}) \bigg] \\
&= \sum_{x_{1:k-1}} \psi_1(x_1) \prod_{t=2}^{k} \psi_{t-1, t}(x_{t-1}, x_{t}) \\
&= \psi_{1,k}^f(x_{k}).
\end{aligned}
\label{eq:theorem1_induction_proof}
\end{equation*}
This concludes the proof.
\end{proof}
\begin{theorem}
The backward potential is given as
\begin{align*}
 a_{k:T+1} &= \psi_{k,T}^b(x_{k}), \quad k > 0.
\end{align*}
\label{th:theorem_backward_potentials}
\end{theorem}

\begin{proof}
We prove by induction that
\begin{equation}
\begin{aligned}
a_{k:T+1} &= a_{k:k+1} \otimes a_{k+1:k+2} \otimes \cdots \\
 &\otimes a_{T-1:T} \otimes a_{T:T+1}  =  \psi_{k,T}^b(x_k)
\end{aligned}
\label{eq:theorem2_induction}
\end{equation}
holds. Eq.~\eqref{eq:theorem2_induction} holds for $k=T$ by definition of $a_{T:T+1} = 1$ in Eq.~\eqref{eq:definition}. Then, we assume that
\begin{equation}
\begin{aligned}
  a_{k+1:T+1} &= a_{k+1:k+2} \otimes a_{k+2:k+3} \otimes \cdots \otimes a_{T:T+1}\\
  &= \psi_{k+1,T}^b(x_{k+1}) 
\end{aligned}
\label{eq:theorem2_induction_case}
\end{equation}
holds. We need to prove that Eq.~\eqref{eq:theorem2_induction} holds for any $k < T$. We start by applying the binary operator $\otimes$ with $a_{k:k+1}$ to the left-hand side of Eq.~\eqref{eq:theorem2_induction_case} to yield
\begin{equation*}
\begin{aligned}
&a_{k:k+1} \otimes a_{k+1:T+1} \\
&=  \sum_{x_{k+1}} \psi_{k,k+1}(x_{k}, x_{k+1}) \, \psi_{k+1, T}^b(x_{k+1}) \\
&= \sum_{x_{k+1}} \bigg[ \psi_{k,k+1}(x_{k}, x_{k+1})\\
&\qquad \times \left( \sum_{x_{k+2:T}} \prod_{t=k+1}^{T} \psi_{t,t+1}(x_{t}, x_{t+1}) \right) \bigg]\\
&=  \sum_{x_{k+1:T}} \prod_{t=k}^{T} \psi_{t,t+1}(x_{t}, x_{t+1})\\
&= \psi_{k,T}^b(x_k).
\end{aligned}
\label{eq:theorem2_induction_proof}
\end{equation*}
This concludes the proof.
\end{proof}
Now, we can express the marginal in Eq.~\eqref{eq:smoothing_dist_extended} using $a_{0:k}$ and $a_{k:T+1}$ as
\begin{equation}
    p(x_k) = \frac{1}{Z_k} \, a_{0:k} \, a_{k:T+1}.
\label{eq:marginals_by_elements}
\end{equation}
In the proofs above we have implicitly used the sum-product algorithm in Algorithm~\ref{alg:fb_algorithm} to derive the results, which can be done in $O(TD^2)$ steps. However, because the operator is associative, we can reorder the computations by recombining the operations which is the key to parallelization. This is discussed next.

\subsection{Parallelization of the sum-product algorithm}
\label{sec:par_sum_product}

In this section, the aim is to enable the parallel computation of the forward and backward potentials. It follows from the previous section that the computation of the forward potentials corresponds to the computation of the all-prefix-sums operation for the associative operator $\otimes$ and elements $a_{i:i+1}$. In fact, we have that
\begin{equation}
\begin{aligned}
  a_{0:1} &= \psi_{1,1}^f(x_1), \\
  a_{0:2} &= a_{0:1} \otimes a_{1:2} = \psi_{1,2}^f(x_2), \\
  &\vdots \\
  a_{0:T} &= a_{0:1} \otimes \cdots  \otimes a_{T-1:T}  =  \psi_{1,T}^f(x_T).
\end{aligned}
\end{equation}
Similarly, computing the backward potentials, which correspond to elements $\{a_{k:T}\}$ can be seen as the reversed all-prefix-sums operation. 

Therefore, we can use the parallel-scan algorithm to compute the forward and backward potentials along with the marginal distributions in parallel. The pseudocode of the resulting algorithm is summarized in Algorithm~\ref{alg:par_sum_prodcut}. 

\begin{algorithm}[H]
\begin{algorithmic}[1]
  \REQUIRE The potentials $\psi_k(\cdot)$, $k=1,\ldots,T$ and the operator $\otimes$, see Definition~\ref{def:filtering_smoothing_operator}.
  \ENSURE The marginals $p(x_k)$ for $k=1,\ldots,T$.
  \FOR[In parallel]{$k\gets1$ \textbf{to} $T$}
  \STATE Initialize $a_{k-1:k}$.
  \ENDFOR
  \STATE Run parallel-scan to get $a_{0:k} = \psi_{1,k}^f(x_{k})$, $k=1,\ldots,T$.
  \FOR[In parallel]{$k\gets1$ \textbf{to} $T$}
  \STATE Initialize $a_{k:k+1}$.
  \ENDFOR
  \STATE Run reversed parallel-scan to get $a_{k:T+1} = \psi_{k,T}^b(x_{k})$, $k=1,\ldots,T$.
  \FOR[In parallel]{$k\gets1$ \textbf{to} $T$}
  \STATE Compute marginals $p(x_k)$ using Eq.~\eqref{eq:marginals_by_elements}.
  \ENDFOR
\end{algorithmic}
\caption{The parallel sum-product algorithm.}
\label{alg:par_sum_prodcut}
\end{algorithm}

As all the steps in Algorithm~\ref{alg:par_sum_prodcut} are either fully parallelizable, or  parallelizable by means of the parallel-scan algorithm, then span and work complexities can be summarized as follows.

\begin{proposition}
The parallel sum-product algorithm (Algorithm~\ref{alg:par_sum_prodcut}) has a span complexity $O(\log T)$ and a work complexity $O(T)$.
\end{proposition}

\begin{proof}
The initializations for the elements for both the forward and backward passes as well as the marginal computations are fully parallelizable and, hence, have span complexities $O(1)$ and work complexities $O(T)$. The parallel-scan algorithm passes have span complexities $O(\log T)$ and work complexities $O(T)$. Therefore, the total span complexity is $O(\log T)$ and the total work complexity is $O(T)$.
\end{proof}

It is useful to remark that the computational complexity depends also on the number of states $D$. However, it depends on the details of how much we can parallelize inside the initializations and operator applications, which is what determines the dependence of the span complexity on $D$.

\section{Parallel Viterbi and max-product algorithms}
\label{sec:max-product}

In this section, we present the parallel formulation of the Viterbi and max-product algorithms. In Section~\ref{subsec:classical_viterbi_algorithm}, we review the classical Viterbi algorithm. In Section~\ref{subsec:path_based_parallelization}, we propose a parallel Viterbi algorithm based on optimal paths. In Section~\ref{subsec:max_product_forumulation}, we propose another alternative parallel Viterbi algorithm based on the max-product formulation.

\subsection{Classical Viterbi algorithm}
\label{subsec:classical_viterbi_algorithm}

The \textit{Viterbi algorithm} is a classical algorithm based on dynamic programming principle that computes the most likely sequence of states~\cite{viterbi1967error,Larson66, cappe2006inference}, also known as the MAP estimate of the hidden states. We aim to compute the estimate $x^*_{1:T}$ by maximizing the posterior distribution
\begin{equation}
\begin{split}
  p(x_{1:T} \mid y_{1:T}) = \frac{p(y_{1:T}, x_{1:T})}{p(y_{1:T})}
  \propto p(y_{1:T},x_{1:T}).
\end{split}    
\end{equation}
This is equivalent to maximizing the joint probability distribution
\begin{equation}
\begin{split}
  p(y_{1:T},x_{1:T}) =
  p(x_1) \, p(y_1 \mid x_1)  \, \prod_{t=2}^T p(y_t \mid x_t) \, p(x_t \mid x_{t-1}),
\end{split}    
\end{equation}
which results in
\begin{equation}
\begin{split}
  x^*_{1:T} = \mathop{\arg \max}_{x_{1:T}} p(y_{1:T},x_{1:T}).
\end{split}
\label{eq:optimal_path}
\end{equation}
In terms of potentials, obtaining $x_{1:T}^*$ is  equivalent to the maximization of Eq.~\eqref{eq:joint_dist_markov}.

The classical Viterbi algorithm \cite{viterbi1967error, cappe2006inference} operates as follows. Assume that we have $V_{k-1}(x_{k-1})$ which denotes the probability of the maximum probability path $x^*_{1:k-1}$ that ends at $x_{k-1}$. We also assume the corresponding optimal state is $x_{k-1}^*$, which is a function of $x_k$, and hence, denoted as $u_{k-1}(x_k)$. Then, $V_k(x_k)$ and $u_{k-1}(x_k)$ are given by
\begin{equation}
\begin{split}
  V_k(x_k) &=
  \max_{x_{k-1}} [p(y_k \mid x_k) \, p(x_k \mid x_{k-1}) \, V_{k-1}(x_{k-1})], \\
  u_{k-1}(x_k) &=
  \mathop{\arg \max}_{x_{k-1}} [p(y_k \mid x_k) \, p(x_k \mid x_{k-1}) \, V_{k-1}(x_{k-1})],
\end{split}    
\end{equation}
with the initial condition
\begin{equation}
\begin{split}
   V_1(x_1) = p(x_1) \, p(y_1 \mid x_1).
\end{split}    
\end{equation}
At the final step, we just take
\begin{equation}
\begin{split}
  x^*_T = \argmax_{x_T} V_T(x_T),
\end{split}    
\end{equation}
and then we compute in backwards to recover the Viterbi path $x_{1:T}^*$ as
\begin{equation}
\begin{split}
  x^*_{k-1} = u_{k-1}(x^*_k).
\end{split}    
\end{equation}

In terms of potentials, the pseudocode for Viterbi algorithm can be written as in Algorithm~\ref{alg:vb_algorithm}. As the computational complexity of the forward pass is $O(D^2 T)$ and that of the backward pass is $O(T)$, the total computational complexity of the Viterbi algorithm is $O(D^2 T)$.

\begin{algorithm}[H]
\begin{algorithmic}[1]
  \REQUIRE The potentials $\psi_k(\cdot)$ for $k=1,\ldots,T$.
  \ENSURE The Viterbi path $x^*_{1:T}$.
  \STATE // Forward pass:
  \STATE $V_1(x_1) = \psi_1(x_1)$
  \FOR[Sequentially]{$k \gets 2$ \textbf{to} $T$}
      \STATE $V_{k}(x_k) = \max_{x_{k-1}} \left[  \psi_k(x_{k-1},x_k) \, V_{k-1}(x_{k-1}) \right]$
      \STATE $u_{k-1}(x_k) = \argmax_{x_{k-1}} \left[  \psi_k(x_{k-1},x_k) \, V_{k-1}(x_{k-1}) \right] $
  \ENDFOR
  \STATE // Backward pass:
  \STATE $x^*_T = \argmax_{x_T} V_T(x_T) $
  
  \FOR{$k \gets T$ \textbf{to} $2$} 
    \STATE $x^*_{k-1} = u_{k-1}(x^*_k)$
  \ENDFOR
\end{algorithmic}
\caption{The classical Viterbi algorithm.}
\label{alg:vb_algorithm}
\end{algorithm}

Please note that, in this paper, we assume that the MAP estimate is unique for simplicity of exposition. However, the results can be extended to account for multiple solutions.

\subsection{Path-based parallelization}
\label{subsec:path_based_parallelization}

In order to design a parallel version of the Viterbi algorithm, we first need to find an element $\tilde{a}$ and the binary associative operator $\vee$ for the Viterbi algorithm. 
Let $\tilde{a}_{i:j}$ consists of a pair of elements associated with probability and the path from state $x_i$ to $x_j$ such that
\begin{equation}
\tilde{a}_{i:j} = \begin{pmatrix}
A_{i:j}(x_i, x_j)\\
\hat{X}_{i:j}(x_i, x_j)
\end{pmatrix}.
\end{equation}
Here, $A_{i:j}(x_i, x_j)$ is the maximum probability of the path and  $\hat{X}_{i:j}(x_i, x_j)$ is a column vector containing the most probable path sequence starting at $x_i$ and ending at $x_j$. Now, we define the associative operator $\vee$.

\begin{definition}
For two elements
\begin{equation}
\tilde{a}_{i:j} = \begin{pmatrix}
A_{i:j}(x_i, x_j)\\
\hat{X}_{i:j}(x_i, x_j)
\end{pmatrix}, \quad
\tilde{a}_{j:k} = \begin{pmatrix}
A_{j:k}(x_j, x_k)\\
\hat{X}_{j:k}(x_j, x_k)
\end{pmatrix},
\end{equation}
the binary associative operator $\vee$ for the MAP estimate in an HMM is defined as 
\begin{equation}
    \begin{split}
  \tilde{a}_{i:k} = \tilde{a}_{i:j} \vee \tilde{a}_{j:k},\\
    \end{split}
\end{equation}
where
\begin{equation}
    \begin{split}
&\tilde{a}_{i:k} = \begin{pmatrix}
A_{i:k}(x_i, x_k)\\
\hat{X}_{i:k}(x_i, x_k)
\end{pmatrix}\\
&= \begin{pmatrix}
\max_{x_j} A_{i:j}(x_i, x_j) A_{j:k}(x_j, x_k)\\
\left( \hat{X}_{i:j}(x_i, \hat{x}_j({x_i, x_k}) ), \hat{x}_j({x_i, x_k}), \hat{X}_{j:k}(\hat{x}_j({x_i, x_k}), x_k) \right)
\end{pmatrix}\\
    \end{split}
\label{eq:viterbi_a_comb}
\end{equation}
and
\begin{equation}
    \begin{split}
& \hat{x}_j({x_i, x_k}) = \argmax_{x_j} A_{i:j}(x_i, x_j) A_{j:k}(x_j, x_k)
    \end{split}
\end{equation}
with
\begin{equation}
    \begin{split}
    A_{k-1:k}(x_{k-1}, x_{k}) &= \psi_{k-1,k}(x_{k-1}, x_{k}), \\
    \hat{X}_{k-1:k}(x_{k-1}, x_{k})  &= \emptyset, \\
    A_{0:1}(., x_1) &= \psi_1(x_1), \\
    \hat{X}_{0:1}(., x_1) &= \emptyset.
    \end{split}
    \label{eq:def_viterbi_first_elem}
\end{equation}
\label{def:viterbi_operator}
\end{definition}

\begin{lemma}
\label{lem:vee_assoc}
The operator $\vee$ is associative.
\end{lemma}

\begin{proof}
See Appendix~\ref{sec:proof_of_assoc_viterbi} for the proof.
\end{proof}

\begin{theorem}
\label{the:path_elem}
We have
\begin{equation}
    \begin{split}
    A_{i:j}(x_{i}, x_{j}) &= \max_{x_{i+1:j-1}} \prod_{k=i+1}^j \psi_{k-1,k}(x_{k-1}, x_{k}), \\
    \hat{X}_{i:j}(x_{i}, x_{j})  &= \argmax_{x_{i+1:j-1}} \prod_{k=i+1}^j \psi_{k-1,k}(x_{k-1}, x_{k}),
    \end{split}
    \label{eq:def_viterbi_elem}
\end{equation}
where $A_{i:j}$ corresponds to the probability of the MAP estimate starting from $x_i$ and ending at $x_j$, and $\hat{X}_{i:j}$ represents the sequence of the corresponding paths.
\end{theorem}

\begin{proof}
See Appendix~\ref{sec:proof_of_theorem_va}.
\end{proof}

Putting $i=0$ and $j=T+1$ in Theorem~\ref{the:path_elem}  above leads to the following result.
\begin{corollary}
We have
\begin{align*}
 \tilde{a}_{0:T+1} &= \begin{pmatrix}
 \psi_1(x^*_1) \prod_{t=2}^T \psi_{t-1, t}(x^*_{t-1}, x^*_t)\\
 x^*_{1:T}
 \end{pmatrix},
\end{align*}
where $x^*_{1:T}$ is the MAP estimate given by Eq.~\eqref{eq:optimal_path}.
\label{cor:theorem_va}
\end{corollary}

It is now possible to form a parallel algorithm for the elements $\tilde{a}_{i:j}$ with the associative operator $\vee$ by leveraging the parallel-scan algorithm for computing the quantity in Corollary~\ref{cor:theorem_va}. However, each element $\tilde{a}_{i:j}$ contains a path of length $j-i-2$ for each state pair and, therefore, the memory requirements to store the state sequences are high. Thus, in the next section, we propose an alternative approach with lower memory requirements.

\subsection{Max-product formulation}
\label{subsec:max_product_forumulation}

Until now, we have used the Viterbi algorithm in forward direction. However, the MAP estimate can also be computed by using the max-product algorithm (see cf.\cite{Koller:2009}). Let $\tilde{\psi}^f_{k}(x_k)$ denote the maximum probability of the optimal path ending at $x_k$, and $\tilde{\psi}^b_{k}(x_k)$ denote the maximum probability of starting at $x_k$. This implies that
\begin{equation}
\begin{split}
 \tilde{\psi}^f_{k}(x_k) &= A_{0:k}(x_0,x_k), \\
 \tilde{\psi}^b_{k}(x_k) &= A_{k:T+1}(x_k,x_{T+1}),
\end{split}
\end{equation}
where the dependence on $x_0$ and $x_{T+1}$ is only a notational expression (cf.\ Eq.~\eqref{eq:definition_potentials}). From Definition~\ref{def:viterbi_operator}, we get the following recursions for these quantities.
\begin{lemma}
The maximum forward and backward probabilities admit the recursions
\begin{equation}
    \begin{aligned}
\tilde{\psi}^f_{k}(x_{k}) &= \max_{x_{k-1}} \psi_{k-1:k}(x_{k-1}, x_k) \tilde{\psi}^f_{k-1}(x_{k-1}), \\ 
    \tilde{\psi}^b_{k}(x_{k}) &= \max_{x_{k+1}} \psi_{k:k+1}(x_{k}, x_{k+1}) \tilde{\psi}^b_{k+1}(x_{k+1}),
    \end{aligned}
\end{equation}
with initial conditions $\tilde{\psi}^f_{1}(x_{1}) = \psi_1(x_1)$ and $\tilde{\psi}^b_T(x_{T}) = 1$.
\end{lemma}
The MAP estimate $x_k^*$ can be computed by maximizing the product $\tilde{\psi}^f_{k}(x_{k}) \tilde{\psi}^b_{k}(x_{k})$. This is summarized in the following theorem.

\begin{theorem}
Given the maximum forward potentials $\tilde{\psi}^f_{k}(x_{k})$ and the maximum backward potentials $\tilde{\psi}^b_{k}(x_{k})$, the MAP estimate at time step k is determined by
\begin{align}
 x^*_k = \argmax_{x_k} \tilde{\psi}^f_{k}(x_{k}) \tilde{\psi}^b_{k}(x_{k})
\label{eq:theorem_va_global_optimum}
\end{align}
\label{th:theorem_va_global_optimum}
for $k=1,\ldots,T$.
\end{theorem}
\begin{proof}
See Appendix~\ref{sec:theorem_va_global_optimum}.
\end{proof}
Theorem~\ref{th:theorem_va_global_optimum} follows from the generalized MAP estimate for an arbitrary graph discussed in Ref.~\cite[ch. 13]{Koller:2009}. However, we compute the elements of these forward and backward potentials in parallel. 

Let us now define an element $\bar{a}_{i:j}$ that consists of the upper part of $\tilde{a}_{i:j}$ where the path is left out. These elements can be computed without actually storing the paths $x^*_{i:j}$ which provides a computational advantage.

\begin{definition}
\label{def:viterbi_operator_for_max_product}
For two elements $\bar{a}_{i:j} = A_{i:j}(x_i, x_j)$ and $\bar{a}_{j:k} = A_{j:k}(x_j, x_k)$, the binary operator $\vee$ can be defined as
\begin{equation}
    \begin{split}
\bar{a}_{i:k} &=  \bar{a}_{i:j} \vee \bar{a}_{j:k},
    \end{split}
\end{equation}
where 
\begin{equation}
    \begin{split}
\bar{a}_{i:k}  = \max_{x_j} A_{i:j}(x_i, x_j) A_{j:k}(x_j, x_{k}).
    \end{split}
\end{equation}
\end{definition}
The element $\bar{a}_{i:j}$ and the operator also inherit the associative property of the element $\tilde{a}_{i:j}$ and, therefore, we also have $\bar{a}_{i:k} = A_{i:k}(x_i, x_k)$.

We can now compute the maximum forward and backward potentials in terms of these associative operators and elements, which is summarized in the following.

\begin{proposition}
The maximum forward potential can be computed as 
\begin{align*}
 \bar{a}_{0:k} &= \tilde{\psi}_{k}^f(x_{k}), \quad k > 0.
\end{align*}
\label{th:theorem_max_forward}
\end{proposition}
\begin{proof}
This follows from Theorem~\ref{the:path_elem}.
\end{proof}
\begin{proposition}
The maximum backward potential can be computed as
\begin{align*}
 \bar{a}_{k:T+1} &= \tilde{\psi}_{k}^b(x_{k}), \quad k > 0.
\end{align*}
\label{th:theorem_max_backward}
\end{proposition}
\begin{proof}
This follows from Theorem~\ref{the:path_elem}.
\end{proof}

Similarly to the sum-product algorithm in Section~\ref{sec:par_sum_product}, we can now parallelize the max-product algorithm. The steps are summarized in  Algorithm~\ref{alg:par_max_prodcut}. The span and work complexities of the algorithm are summarized in the following proposition.

\begin{algorithm}[H]
\begin{algorithmic}[1]
  \REQUIRE The potentials $\psi_k(\cdot)$, $k=1,\ldots,T$ and the operator $\vee$, see Definition~\ref{def:viterbi_operator_for_max_product}.
  \ENSURE The Viterbi path $x^*_{1:T}$.
  \FOR[In parallel]{$k\gets1$ \textbf{to} $T$}
  \STATE Initialize $\bar{a}_{k-1:k}$.
  \ENDFOR
  \STATE Run parallel-scan to get $\bar{a}_{0:k} = \tilde{\psi}_{1,k}^f(x_{k})$, $k=1,\ldots,T$.
  \FOR[In parallel]{$k\gets1$ \textbf{to} $T$}
  \STATE Initialize $\bar{a}_{k:k+1}$.
  \ENDFOR
  \STATE Run reversed parallel-scan to get $\bar{a}_{k:T+1} = \tilde{\psi}_{k,T}^b(x_{k})$, $k=1,\ldots,T$.
  \FOR[In parallel]{$k\gets1$ \textbf{to} $T$}
  \STATE Compute the optimal state $x_k^*$ using Eq.~\eqref{eq:theorem_va_global_optimum}.
  \ENDFOR
\end{algorithmic}
\caption{The parallel max-product algorithm.}
\label{alg:par_max_prodcut}
\end{algorithm}

\begin{proposition}
The span complexity of the parallel max-product algorithm (Algorithm~\ref{alg:par_max_prodcut}) is $O(\log T)$ and the work complexity is $O(T)$.
\end{proposition}

\begin{proof}
The initializations and final state combinations have span complexities $O(1)$ and work complexities $O(T)$, whereas the parallel-scans have span complexities $O(\log T)$ and work complexities $O(T)$. Hence the result follows.
\end{proof}

\section{Extensions}
\label{sec:extensions}

In this section, we discuss extensions of the parallel-scan framework for the inference in HMMs.
\subsection{Generic associative operations}

We defined the parallel-scan framework in terms of sum-product and max-product for HMMs. We can easily extend this framework to operations of the form expressed by Eq.~\eqref{eq:general_operation} for arbitrary associative operators. In particular, we can also consider continuous-state Markov processes; in this case, the operator becomes integration and we get similar algorithms to the ones described in~\cite{sarkka2020temporal}, except that the smoother will be a two-filter smoother instead of a Rauch--Tung--Striebel smoother. In particular, for linear Gaussian systems, we get a parallel version of the two-filter Kalman smoother.

\subsection{Block-wise operations}
In this article, we have restricted our consideration to pair-wise Markovian elements, that is, we assign a single observation and state to a single element. However, it is possible to perform binary association operations in parallel on a block level, where we assign a single computational element to a set of consecutive observations and states~\cite{sarkka2020temporal}. In other words, we can define a block of consecutive $l$ observations and states as a single element in the parallel framework. This single element processes a block of measurements before combining the results with other elements. This kind of block-processing can be advantageous when the number of computational cores is limited.

\subsection{Parameter estimation}
Another task of the HMM inference is to estimate parameters of the model. One possible solution is to use the \textit{Baum-Welch algorithm} (BWA)~\cite{baum1970maximization}, which is a special case of \textit{expectation--maximization} (EM) algorithm~\cite{dempster1977maximum}. In expectation step, BWA uses the forward-backward algorithm, which can be parallelized using the methods proposed in this article.

\section{Experimental results}
\label{sec:experimental_results}

For our experiment, we consider the Gilbert--Elliott (GE) channel model~\cite{cappe2006inference}, which is a classical model used in the transmission of signals in digital communication channels. The model describes the burst error patterns in communication channels and simulates the error performance of the communication link. This model consists of two hidden Markov states $b_k$ and $s_k$, where $b_k$ represents the binary input signal to be transmitted and $s_k$ represents the channel regime binary signal. It is assumed that the input signal $b_k$ may be flipped by an independent error, which can be modeled as $y_k = b_k \oplus v_k$. Here, $y_k$ is the measurement signal, which is observable. Furthermore, $v_k \in \{0, 1\}$ is a Bernoulli sequence, and $\oplus$ is the exclusive-or operation. The exclusive-or operation ensures that $y_k = b_k$ if $v_k = 0$, and $y_k \neq b_k$ if $v_k = 1$.

The regime input signal $s_k$ is modeled as a two-state Markov chain that represents high and low error conditions of the channel. More specifically, if an error occurs ($v_k = 1$), the probability of error has either a small value ($q_0$) or a large value ($q_1$). Moreover, we denote the probability of transition from a high error state ($s_{k-1} = 1$) to a low error state ($s_{k} = 0$) as $p_0$. Similarly, we present the probability of transition from a low error state ($s_{k-1} = 0$) to a high error state ($s_{k} = 1$) as $p_1$. We also denote the state switch probability of $b_k$ as $p_2$.

\begin{figure}
\centerline{\includegraphics[width=\columnwidth]{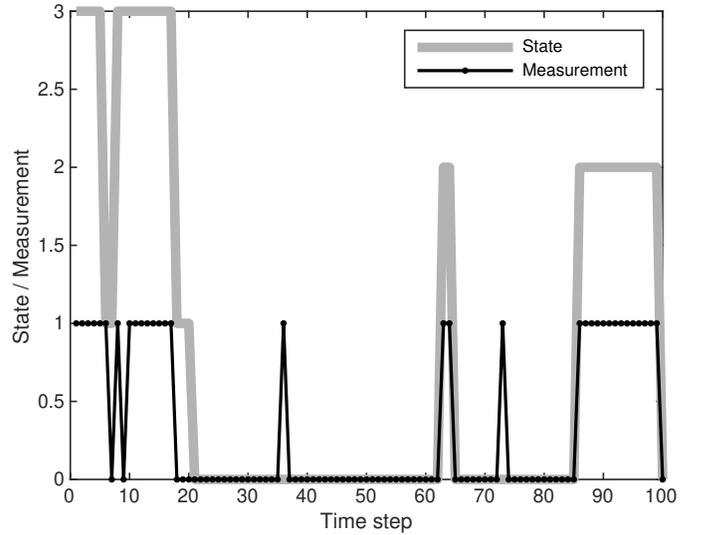}}
\caption{An example of states and measurements from the Gilbert--Elliot hidden Markov model with the number of time steps $T=100$.}
\label{fig:ex_data}
\end{figure}

\begin{figure*}[!ht]
\normalsize
\begin{equation}
\begin{split}
  \Pi &= \begin{pmatrix}
       (1-p_0) \, (1-p_2) & p_0 \, (1-p_2) & (1-p_0) \, p_2 & p_0 \, p_2 \\
           p_1 \, (1-p_2) & (1-p_1) \, (1-p_2) & p_1 \, p_2 & (1-p_1) \, p_2 \\
          (1-p_0) \, p_2 & p_0 \, p_2 & (1-p_0) \, (1-p_2) & p_0 \, (1-p_2) \\
           p_1 \, p_2 & (1-p_1) \, p_2 & p_1 \, (1-p_2) & (1-p_1) \, (1-p_2)
           \end{pmatrix}, \quad
  O = \begin{pmatrix}
           (1-q_0) & q_0 \\
           (1-q_1) & q_1 \\
            q_0  & (1-q_0) \\
            q_1  & (1-q_1) \\
           \end{pmatrix}.
\end{split}
\label{eq:gilbert_elliot_priors}
\end{equation}
Here, 
\begin{equation*}
\Pi = p(x_k \mid x_{k-1})    
\end{equation*}
and
\begin{equation*}
O = p(y_k \mid x_{k}).    
\end{equation*}
\hrulefill
\vspace*{4pt}
\end{figure*}

In order to recover the hidden states, we use the joint model $x_k = (s_k, b_k)$ which is a 4-state Markov chain ($D=4$) consisting of the states $\{(0, 0), (0, 1), (1, 0), (1, 1)\}$. These states are encoded as $x_k \in \{0, 1, 2, 3\}$. The transition matrix $\Pi$ and the observation model $O$, which encode the information in $p(x_k \mid x_{k-1})$ and $p(y_k \mid x_{k})$, respectively, are given in Eq.~\eqref{eq:gilbert_elliot_priors}. An example of the states and the corresponding measurements from the GE model is shown in Fig.~\ref{fig:ex_data}.

To evaluate the performance of our proposed methods, we simulated the states and measurements with varying lengths of time series ranging from $T=10^2$ to $T=10^5$ and averaged the run times (10 repetitions for sequential methods and 100 for parallel ones). It is worth noting that we evaluate the methods only in terms of computational time (not error performance), because the parallel and sequential methods are algebraically equivalent and, therefore, there is no difference in their error performance.

In the experimental setup, we used the open-source TensorFlow 2.4 software library \cite{tensorflow2015-whitepaper}. The library natively implements the vectorization and associative scan primitives that can be used to implement the methodology for both CPUs and GPUs. In this experiment, we set the values of parameters of the GE model given in Eq.~\eqref{eq:gilbert_elliot_priors} as follows: $p_0 = 0.03, p_1 = 0.1, p_2 = 0.05, q_0 = 0.01, q_1 = 0.1$. The initial prior for the state is uniform, $p(x_1) = 0.25$ for $ x_1 \in  \{0, 1, 2, 3\}$. We ran the sequential and parallel Bayesian smoothers (BS-Seq, BS-Par)~\cite{sarkka2013bayesian}, sequential and parallel sum-product based smoothers  (SP-Seq, SP-Par) from Section~\ref{sec:sum-product}, and sequential and parallel max-product based MAP estimators (MP-Seq, MP-Par) from Section~\ref{sec:max-product} on both CPU and GPU. We also ran the classical Viterbi algorithm (see Algorithm~\ref{alg:vb_algorithm}) for comparison. The experiment was carried both on a CPU, AMD Ryzen$\textsuperscript{TM}$ Threadripper$\textsuperscript{TM}$ 3960X with 24 Cores and 3.8 GHz, and a GPU, NVIDIA$\textsuperscript{\textregistered}$ Ampere$\textsuperscript{\textregistered}$ GeForce RTX 3090 (GA102) with 10496 cores.

\begin{figure}
\centerline{\includegraphics[width=\columnwidth]{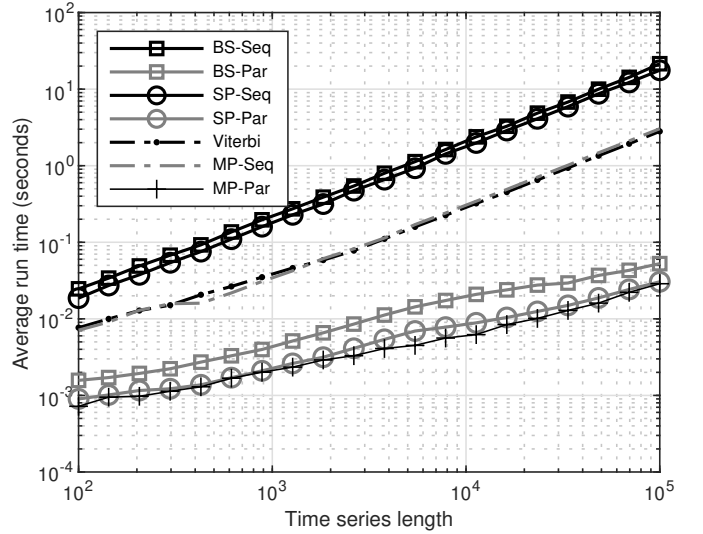}}
\caption{Average computation times on the CPU for sequential and parallel Bayesian smoothers (BS-Seq, BS-Par), sequential and parallel sum-product algorithms (SP-Seq, SP-Par), sequential and parallel max-product algorithms (MP-Seq, MP-Par), and the classical Viterbi algorithm (Viterbi).}
\label{fig:cpu}
\end{figure}

\begin{figure}
\centerline{\includegraphics[width=\columnwidth]{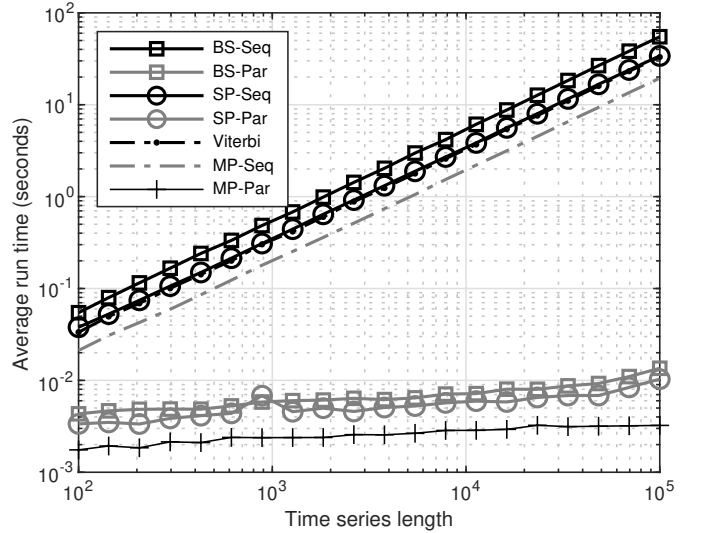}}
\caption{Average computation times on the GPU for sequential and parallel Bayesian smoothers (BS-Seq, BS-Par), sequential and parallel sum-product algorithms (SP-Seq, SP-Par), sequential and parallel max-product algorithms (MP-Seq, MP-Par), and the classical Viterbi algorithm (Viterbi).}
\label{fig:gpu}
\end{figure}

Average run times are shown in Figs.~\ref{fig:cpu} and~\ref{fig:gpu}. It is clear that the parallel algorithms are computationally faster than the sequential versions both on the CPU and the GPU. Although the difference is much more pronounced on the GPU, even on the CPU, we can see a benefit of parallelization. Nevertheless, the number of computational cores is orders of magnitude smaller than in the GPU (24 vs.\ 10496 cores). Among the compared methods, the max-product-based parallel method is the fastest, sum-product-based parallel method is the second, and the parallel Bayesian smoother is the third, on both the CPU and the GPU. A similar order can be seen among the sequential method results and the classical Viterbi is placed between the other sequential methods.

\begin{figure}
\centerline{\includegraphics[width=\columnwidth]{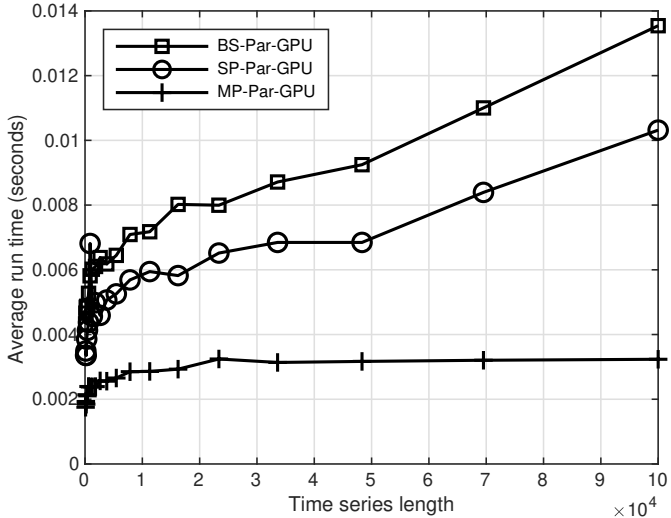}}
\caption{Average computation times on the GPU for parallel Bayesian smoothers (BS-Par-GPU), parallel sum-product algorithm (SP-Par-GPU), and parallel max-product algorithm (MP-Par-GPU).}
\label{fig:gpu_par}
\end{figure}

Average run times of the parallel methods on the GPU (on a linear scale) are separately shown in Fig.~\ref{fig:gpu_par}. From the figure, it can be seen that the computational times initially grow logarithmically as is predicted by the theory and then, for the Bayesian smoother and the sum-product method retain back to linear when the time series length becomes longer than $\sim5 \times 10^4$. However, the max-product method remains sub-linear even beyond the time series length of $10^5$, which is likely due to less computational operations needed. It is worth mentioning that the mean absolute error between Bayesian smoothers and sum-product based smoothers is insignificant $(\leq10^{-16})$. This difference is due to numerical inaccuracies. The same conclusion can be drawn for max-product based MAP estimators and the Viterbi algorithm.

\begin{figure}
\centerline{\includegraphics[width=\columnwidth]{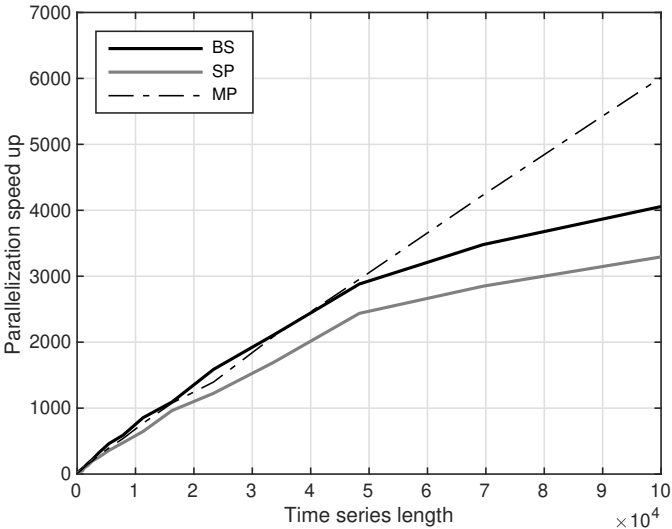}}
\caption{Ratio of the average run times of the sequential methods (BS) and the corresponding parallel methods (SP, MP) on the GPU.}
\label{fig:speedup}
\end{figure}

Finally, Fig.~\ref{fig:speedup} shows the ratio of the average run times of the sequential methods and the corresponding parallel methods on the GPU. It can be seen that with time series length of $\sim5 \times 10^4$, the speed-up is already between 2000--3000 and with time series of length $10^5$, the speed-up in the max-product method is $\sim6000$. On the other hand, with the Bayesian smoother and sum-product methods, the speed-up is $\sim$ 3000--4000.

\section{Conclusion}
\label{sec:conclusion}

In this paper, we have proposed a parallel formulation of HMM inference. We have considered the computation of both the marginal smoothing distributions as well as the MAP estimate. The proposed formulation enables efficient parallel computation is HMMs by reducing the time complexity from linear to logarithmic. The algorithms are based on reformulating the HMM inference problems in terms of associative operators which can be parallelized using the parallel-scan algorithm. We also showed the practical advantage of our proposed methods with experiments in multi-core CPU and GPU.

\appendix

\subsection{Proof of Lemma~\ref{lem:otimes_assoc}}
\label{sec:proof_of_assoc_fs}
In this section, we prove the associative property of the operator $\otimes$ stated in Lemma~\ref{lem:otimes_assoc}. For this, we need to prove that for three general elements $a_{i:j}$, $a_{j:k}$, $a_{k:l}$, the following statement holds
\begin{equation}
\begin{aligned}
    &\left(a_{i:j} \,\otimes\,  a_{j:k}\right)\, \otimes \, a_{k:l}  = a_{i:j}\, \otimes\,  \left(a_{j:k} \otimes a_{k:l}\right), \\
    &\quad \text{where} \, 0 \leq i < j < k < l.
    \end{aligned}
    \label{eq:associative_property_filtering_smoothing}
\end{equation}
We proceed to perform the calculations to the left-hand side of Eq.~\eqref{eq:associative_property_filtering_smoothing}
to check that they yield the same result as in the right-hand side, that is,
\begin{equation*}
\begin{aligned}
    &\left(a_{i:j} \,\otimes\,  a_{j:k}\right)\, \otimes \, a_{k:l}  \\
    &= \sum_{x_{k}} \left( \sum_{x_{j}} \psi_{i, j}(x_{i}, x_{j}) \, \psi_{j,k}(x_{j}, x_{k})\right) \, \psi_{k, l}(x_{k}, x_{l}) \\
    &= \sum_{x_{j}, x_{k}} \psi_{i, j}(x_{i}, x_{j}) \, \psi_{j,k}(x_{j}, x_{k}) \, \psi_{k, l}(x_{k}, x_{l}) \\
    &= \sum_{x_{j}} \psi_{i, j}(x_{i}, x_{j}) \, \left( \sum_{x_{k}} \psi_{j,k}(x_{j}, x_{k}) \, \psi_{k, l}(x_{k}, x_{l})\right) \\
    &= a_{i:j}\, \otimes\,  \left(a_{j:k} \otimes a_{k:l}\right),
\end{aligned}
\label{eq:associative_property_filtering_smoothing_proof}
\end{equation*}
which gives the result.

\subsection{Proof of Lemma~\ref{lem:vee_assoc}}
\label{sec:proof_of_assoc_viterbi}
In this section, we prove the associative property of the operator $\vee$ as stated in Lemma~\ref{lem:vee_assoc}. We need to prove that for three general elements $\tilde{a}_{i:j}$, $\tilde{a}_{j:k}$, $\tilde{a}_{k:l}$, the following statement holds
\begin{equation}
\begin{aligned}
    &(\tilde{a}_{i:j} \vee \tilde{a}_{j:k}) \vee \tilde{a}_{k:l} = \tilde{a}_{i:j} \vee (\tilde{a}_{j:k} \vee \tilde{a}_{k:l}),\\
    &\quad \text{where} \, 0 \leq i < j < k < l.
    \end{aligned}
    \label{eq:associative_property_va}
\end{equation}
From Definition~\ref{def:viterbi_operator}, we can write
\begin{equation}
    \tilde{a}_{i:j} \vee \tilde{a}_{j:k} = \begin{pmatrix}
    A_{i:k}(x_{i}, x_{k})\\
    \hat{X}_{i:k}(x_{i}, x_{k})
    \end{pmatrix},
\end{equation}
where
\begin{subequations}
\begin{align}
    A_{i:k}(x_{i}, x_{k}) &= \max_{x_j} A_{i:j}(x_{i}, x_{j}) A_{j:k}(x_{j}, x_{k}), \label{eq:va_lhs_A}\\
    \hat{X}_{i:k}(x_{i}, x_{k}) &= \big( \hat{X}_{i:j}(x_{i},  \hat{x}_{j}(x_{i}, x_j) ), \nonumber\\
    & \hat{x}_{j}(x_{i}, x_{j}), \hat{X}_{j:k}( \hat{x}_{j}(x_{i}, x_{j}), x_{k})) \big), \label{eq:va_lhs_X}
    \end{align}
\end{subequations}
and
\begin{equation}
    \hat{x}_{j}(x_{i}, x_{j}) = \argmax_{x_j} A_{i:j}(x_{i}, x_{j}) A_{j:k}(x_{j}, x_{k}).
\end{equation}

Now, combining the maximum probability of MAP estimates for the element $\tilde{a}_{k:l}$, we can write Eq.~\eqref{eq:va_lhs_A} as
\begin{equation}
\begin{aligned}
A_{i:l}(x_{i}, x_{l}) &= \max_{x_k} A_{i:k}(x_{i}, x_{k}) A_{k:l}(x_{k}, x_{l}) \\
&= \max_{x_k} \bigg[ \left\{\max_{x_j} A_{i:j}(x_{i}, x_{j}) A_{j:k}(x_{j}, x_{k})\right\} \\
&\quad \times A_{k:l}(x_{k}, x_{l}) \bigg] \\
&= \max_{x_j} \bigg[ A_{i:j}(x_{i}, x_{j}) \\
&\quad \times \left\{ \max_{x_k} A_{j:k}(x_{j}, x_{k}) A_{k:l}(x_{k}, x_{l})  \right\} \bigg] \\
&= \max_{x_j} A_{i:j}(x_{i}, x_{j}) A_{j:l}(x_{j}, x_{l}). \\
\end{aligned}
\label{eq:va_rhs_A}
\end{equation}
Now, combining the Viterbi path of the element $\tilde{a}_{k:l}$, we can write Eq.~\eqref{eq:va_lhs_X} as
\begin{equation}
\begin{aligned}
&\hat{X}_{i:l}(x_{i}, x_{l}) = \argmax_{x_k} A_{i:k}(x_{i}, x_{k}) A_{k:l}(x_{k}, x_{l}) \\
&= \argmax_{x_k} \bigg[ \left\{\argmax_{x_j} A_{i:j}(x_{i}, x_{j}) A_{j:k}(x_{j}, x_{k})\right\} \\
&\quad \times A_{k:l}(x_{k}, x_{l}) \bigg] \\
&= \argmax_{x_j} \bigg[ A_{i:j}(x_{i}, x_{j}) \\
&\quad \times \left\{ \argmax_{x_k} A_{j:k}(x_{j}, x_{k}) A_{k:l}(x_{k}, x_{l})  \right\} \bigg] \\
&= \argmax_{x_j} A_{i:j}(x_{i}, x_{j}) A_{j:l}(x_{j}, x_{l}). \\
\end{aligned}
\label{eq:va_rhs_X}
\end{equation}
From Eq.~\eqref{eq:va_rhs_A} and Eq.~\eqref{eq:va_rhs_X},
it follows that the operator $\vee$ is associative.
\subsection{Proof of Theorem~\ref{the:path_elem}}
\label{sec:proof_of_theorem_va}
In this section, we prove Theorem~\ref{the:path_elem} by induction. We first note that the theorem is true for $i=k$ and $j=k+1$. Furthermore, if the assertion is true for $i+1 < j$, we proceed to show that it holds for any $i$. By using Eq.~\eqref{eq:viterbi_a_comb} we get
\begin{equation}
    \begin{split}
    &A_{i:j}(x_{i}, x_{j}) \\ &=
    \max_{x_{i+1}} A_{i:i+1}(x_i, x_{i+1}) 
    A_{i+1:j}(x_{i+1}, x_j) \\
    &= \max_{x_{i+1}} \psi_{i,i+1}(x_{i}, x_{i+1})
    \max_{x_{i+2:j-1}} \prod_{t=i+2}^j \psi_{t-1,t}(x_{t-1}, x_{t}) \\
    &= \max_{x_{i+1:j-1}} \prod_{t=i+1}^j \psi_{t-1,t}(x_{t-1}, x_{t}).
    \end{split}
\end{equation}
We similarly get
\begin{equation}
    \begin{split}
    &\hat{X}_{i:j}(x_i, x_j) \\
    &= \bigg( \hat{X}_{i:i+1}(x_i, \hat{x}_{i+1}({x_i, x_j})), \hat{x}_{i+1}({x_i, x_j}), \\ &\qquad \hat{X}_{i+1:j}(\hat{x}_{i+1}({x_i, x_j}), x_j) \bigg) \\
    &= \left( \hat{x}_{i+1}({x_i, x_j}), \hat{X}_{i+1:j}(\hat{x}_{i+1}({x_i, x_j}), x_j) \right) \\
    &= \bigg( \argmax_{x_{i+1}} \psi_{i,i+1}(x_{i}, x_{i+1})
    \max_{x_{i+2:j-1}} \prod_{t=i+2}^j \psi_{t-1,t}(x_{t-1}, x_{t}), \\
    &\qquad \max_{x_{i+1}} \argmax_{x_{i+2:j-1}} \prod_{t=i+2}^j \psi_{t-1,t}(x_{t-1}, x_{t})
    \bigg) \\
    &= \argmax_{x_{i+1:j-1}} \prod_{t=i+1}^j \psi_{t-1,t}(x_{t-1}, x_{t}),
    \end{split} 
\end{equation}
which concludes the proof.

\subsection{Proof of Theorem~\ref{th:theorem_va_global_optimum}}
\label{sec:theorem_va_global_optimum}
In this section, we prove  Theorem~\ref{th:theorem_va_global_optimum}. That is, given the maximum forward probability $\tilde{\psi}^f_{k}(x_{k})$ of path ending at $x_{k}$, and the maximum backward probability $\tilde{\psi}^b_{k}(x_{k})$ of path starting at $x_k$, the MAP estimate at time step k is given by Eq.~\eqref{eq:theorem_va_global_optimum}.

Due to the associative property of the operator, we can  write
\begin{equation}
    \tilde{a}_{0:T+1}
    = \tilde{a}_{0:k} \vee \tilde{a}_{k:T+1}
    = \begin{pmatrix}
 \psi_1(x^*_1) \prod_{t=2}^T \psi_t(x^*_{t-1}, x^*_t)\\
 x^*_{1:T}
 \end{pmatrix},
\label{eq:a0Tp1}
\end{equation}
which is given by
\begin{equation}
    \begin{split}
&\tilde{a}_{0:k} \vee \tilde{a}_{k:T+1} \\
&= \begin{pmatrix}
\max_{x_k} A_{0:k}(x_0, x_k) A_{k:T+1}(x_k, x_{T+1})\\
\begin{split}
&\bigg( \hat{X}_{0:k}(x_0, \hat{x}_k({x_0, x_{T+1}}) ), \hat{x}_k({x_0, x_{T+1}}), \\
&\qquad \qquad \hat{X}_{k:T+1}(\hat{x}_k({x_0, x_{T+1}}), x_{T+1}) \bigg)
\end{split}
\end{pmatrix}.
    \end{split}
\end{equation}
Here, the dependence on $x_0$ and $x_{T+1}$ is only a notational expression and the term
\begin{equation}
    \begin{split}
& \hat{x}_k({x_0, x_{T+1}}) \\
&= \argmax_{x_k} A_{0:k}(x_0, x_k) A_{k:T+1}(x_k, x_{T+1}) \\
&= \argmax_{x_k} \tilde{\psi}^f(x_k) \, \tilde{\psi}^b(x_k)
    \end{split}
\end{equation}
has to be the $k$th element on the optimal path in order to match the right-hand side of Eq.~\eqref{eq:a0Tp1}.

\section*{Acknowledgment}
The authors would like to thank Academy of Finland for funding. The authors would also like to thank Adrien Corenflos for helping with the experiments.

\bibliographystyle{IEEEtran}
\bibliography{main}

\vfill

\begin{IEEEbiography}[{\includegraphics[width=1in,height=1.25in,clip,keepaspectratio]{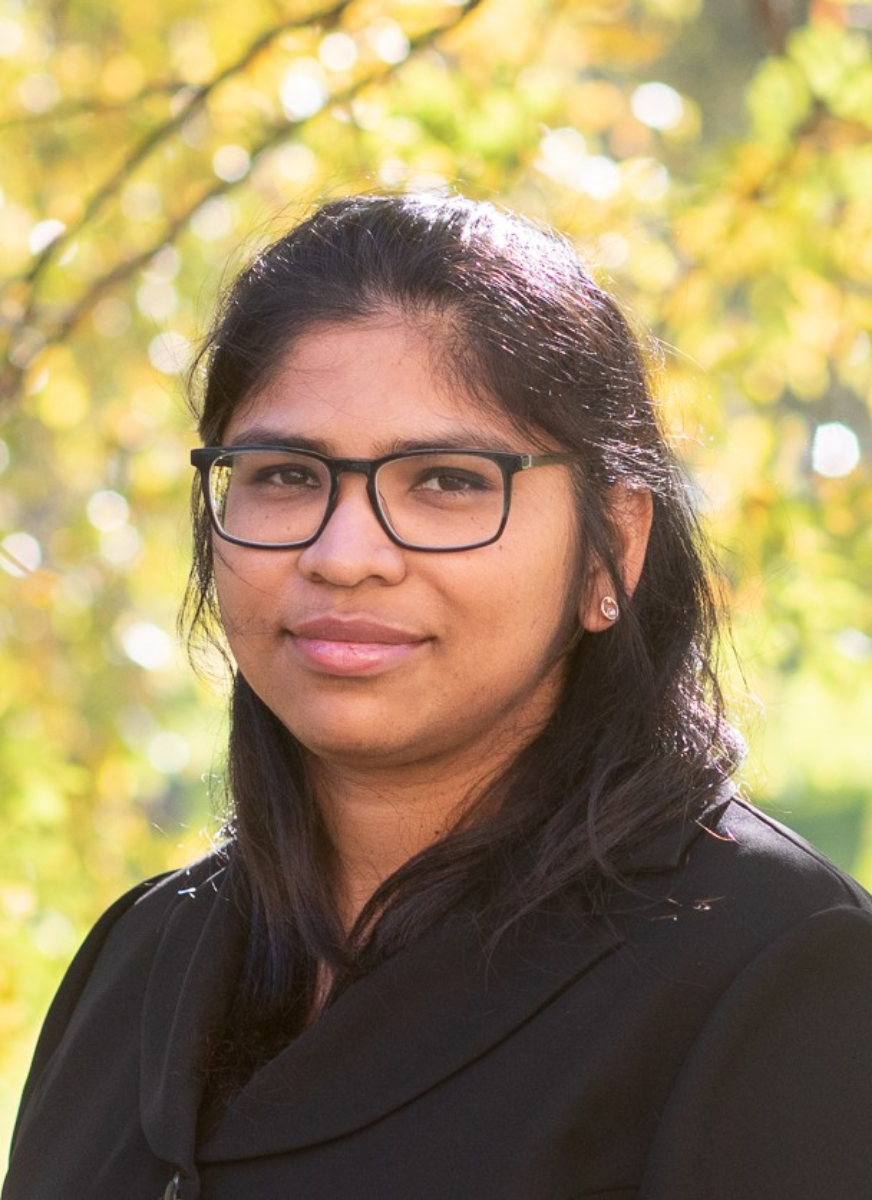}}]{Syeda Sakira Hassan} received her Master of Science (Tech.) degree (with distinction) in information technology, and Doctor of Science (Tech.) degree in computing and electrical engineering from Tampere University, Tampere, Finland, in 2013 and 2019, respectively. From 2019 to 2020 she worked in Basware Oyj, Finland as a Data Scientist. Currently, Dr. Hassan is a Postdoctoral Researcher in the Sensor Informatics and Medical Technology research group at Department of Electrical Engineering and Automation (EEA), Aalto University. She is also working as an AI Scientist with Silo.AI, Finland. Her research interests are machine learning, probabilistic modeling, dynamic programming, and optimal control theory.
\end{IEEEbiography}

\begin{IEEEbiography}[{\includegraphics[width=1in,height=1.25in,clip,keepaspectratio]{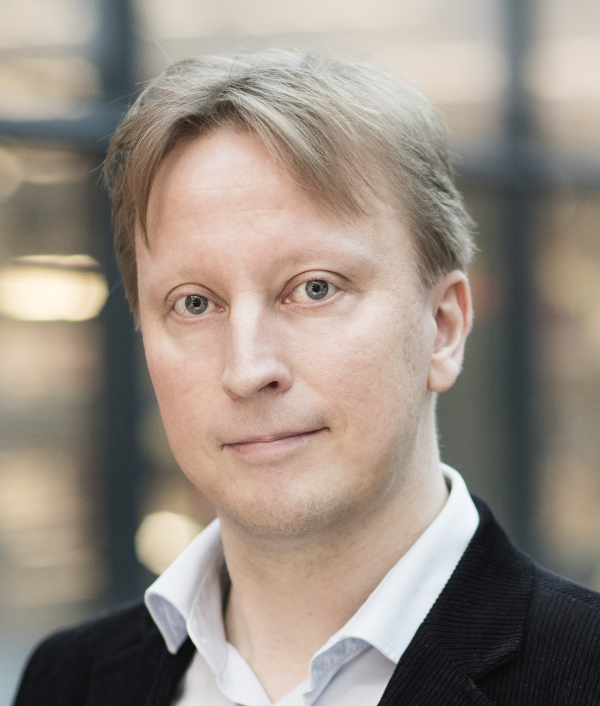}}]{Simo S\"arkk\"a} received his Master of Science (Tech.) degree (with distinction) in engineering physics and mathematics, and Doctor of Science (Tech.) degree (with distinction) in electrical and communications engineering from Helsinki University of Technology, Espoo, Finland, in 2000 and 2006, respectively. Currently, Dr. S\"arkk\"a is an Associate Professor with Aalto University, Technical Advisor of IndoorAtlas Ltd., and an Adjunct Professor with Tampere University of Technology and Lappeenranta University of Technology. His research interests are in multi-sensor data processing systems with applications in location sensing, health and medical technology, machine learning, inverse problems, and brain imaging. He has authored or coauthored over 150 peer-reviewed scientific articles and his books "Bayesian Filtering and Smoothing" and "Applied Stochastic Differential Equations" along with the Chinese translation of the former were recently published via the Cambridge University Press. He is a Senior Member of IEEE and serving as an Associate Editor of IEEE Signal Processing Letters.
\end{IEEEbiography}

\begin{IEEEbiography}[{\includegraphics[width=1in,height=1.25in,clip,keepaspectratio]{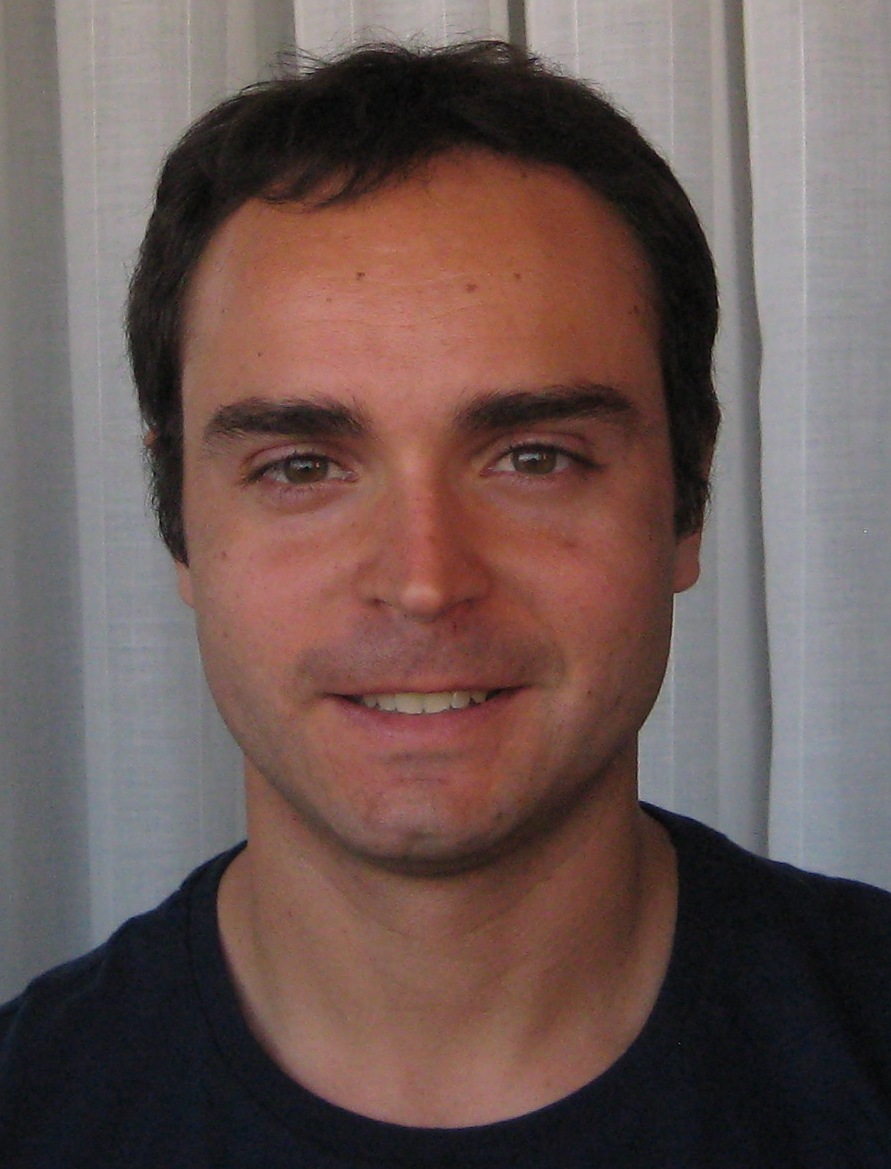}}]{\'Angel F. Garc\'ia-Fern\'andez} received the telecommunication engineering degree (with honours) and the Ph.D. degree from Universidad Polit\'ecnica de Madrid, Madrid, Spain, in 2007 and 2011, respectively. He is currently a Lecturer in the Department of Electrical Engineering and Electronics at the University of Liverpool, Liverpool, UK. His main research activities and interests are in the area of Bayesian inference, with emphasis on nonlinear dynamic systems and multiple target tracking. He has received paper awards at the International Conference on Information Fusion in 2017 and 2019.
\end{IEEEbiography}

\end{document}